\documentclass{article}

\usepackage[utf8]{inputenc}
\usepackage[OT1]{fontenc}
\usepackage[ocgcolorlinks,colorlinks=true,urlcolor=blue]{hyperref}
\usepackage[usenames,dvipsnames]{color}

\usepackage[a4paper,twoside]{geometry}
\usepackage{RR}

\usepackage[usenames,dvipsnames]{color}

\usepackage{amsmath,amssymb,enumerate}

\usepackage{amsmath,amssymb,amsthm,mathabx}
\newtheorem{theorem}{Theorem}
\newtheorem{lemma}[theorem]{Lemma}

\usepackage{bbold}
\usepackage{stmaryrd}

\newcommand{\R}{\mathbb{R}}
\newcommand{\N}{\mathbb{N}}
\newcommand{\pr}{\mathcal{P}}

\newcommand{\pth}[1]{\left( #1 \right)}
\newcommand{\Sp}{S\setminus\{q\}}

\def\cI{\mathcal{I}}

\RRNo{8154}
\RRdate{December 2012}
\RRauthor{

Olivier Devillers\thanks{Projet Geometrica, INRIA Sophia Antipolis - Méditerranée}
\and Marc Glisse\thanks{Projet Geometrica, INRIA Saclay - Île de France}
\and Xavier Goaoc\thanks{Projet Vegas, INRIA Nancy - Grand Est}
\and \\Guillaume Moroz\thanks{Projet Vegas, INRIA Nancy - Grand Est}
\and Matthias Reitzner\thanks{Institut f\"ur Mathematik, Universit\"at Osnabr\"uck, 49069 Osnabr\"uck, Germany.}
}
\authorhead{O. Devillers, M. Glisse, X. Goaoc, G. Moroz, \& M. Reitzner}
\RRtitle{Monotonie des  $f$-vecteurs de polytopes aléatoires}
\RRetitle{The monotonicity of $f$-vectors of random  polytopes}
\titlehead{The monotonicity of $f$-vectors of random  polytopes}
\RRnote{Part of this work is supported by: ANR blanc PRESAGE (ANR-11-BS02-003)}
\RRresume{
  Soit  $K$ un domaine convexe et compact de $\R^d$,  $K_n$ 
  l'enveloppe convexe de $n$ points choisis uniformément et
  indépendamment dans $K$, et
  $f_{i}(K_n)$ le nombre de faces de $K_n$ de dimension $i$.

Nous montrons que pour des convexes du plan, $E[f_0 (K_n)]$
est croissant avec $n$.  En dimension $d \geq 3$ nous montrons que si
$\lim_{n \to \infty}  \frac{E[f_{d-1}(K_n)]}{An^c}=1$ pour des constantes $A$ et $c>0$
alors la fonction $n \mapsto E[f_{d-1}(K_n)]$ est
croissante pour $n$ suffisamment grand. En particulier, le nombre de 
facettes de l'enveloppe convexe de $n$ points aléatoires distribués
uniformément et  indépendamment dans un convexe compact à bord lisse
est asymptotiquement croissant. Notre démonstration utilise un
argument d'échantillonnage aléatoire.
}
\RRabstract{
  Let $K$ be a compact convex body in $\R^d$, let $K_n$ be the convex
  hull of $n$ points chosen uniformly and independently in $K$, and let
  $f_{i}(K_n)$ denote the number of $i$-dimensional faces of $K_n$.

  We show that for planar convex sets, $E[f_0 (K_n)]$ is increasing in
  $n$.  In dimension $d \geq 3$ we prove that if $\lim_{n \to \infty}
  \frac{E[f_{d-1}(K_n)]}{An^c}=1$ for some constants $A$ and $c>0$
  then the function $n \mapsto E[f_{d-1}(K_n)]$ is increasing for $n$
  large enough. In particular, the number of facets of the convex hull
  of $n$ random points distributed uniformly and independently in a
  smooth compact convex body is asymptotically increasing. Our proof
  relies on a \emph{random sampling} argument.
}
\RRmotcle{Géométrie algorithmique, Géométrie stochastique, Enveloppe
  convexe, Complexité.}
\RRkeyword{Computational geometry, Stochastic geometry, Convex hull, Complexity .}
\RRprojets{Geometrica and Vegas}
 \RCSaclay 
 \RCNancy 
\RCSophia 

\begin{document}
\makeRR  

\section{Introduction}

What does a random polytope, that is, the convex hull of a finite set
of random points in $\R^d$, look like? This question goes back to
Sylvester's \emph{four point problem}, which asked for the probability
that four points chosen at random be in convex position.  There are,
of course, many ways to distribute points randomly, and as with
Sylvester's problem, the choice of the distribution drastically
influences the answer.

\paragraph{Random polytopes.}

In this paper, we consider a \emph{random polytope} $K_n$ obtained as
the convex hull of $n$ points distributed uniformly and independently
in a \emph{convex body} $K \subset \R^d$, i.e. a compact convex set with nonempty interior. This model arises
naturally in applications areas such as computational geometry~\cite{Ed, PS}, convex geometry and stochastic geometry~\cite{SchnWe3} or functional analysis~\cite{DGG, MilPa}. A natural question is to understand the behavior
of the {\boldmath $f$}-\emph{vector} of $K_n$, that is, {\boldmath $f$}$(K_n)=(f_0(K_n), \ldots,f_{d-1}(K_n))$ 
where $f_i(K_n)$ counts the number of
$i$-dimensional faces, and the behaviour of the \emph{volume} $V(K_n)$ . Bounding $f_i(K_n)$ is related, for example, to
the analysis of the computational complexity of algorithms in
computational geometry.

Exact formulas for the expectation of {\boldmath $f$}$(K_n)$ and $V(K_n)$ are only 
known for convex polygons ~\cite[Theorem~5]{Buchta-Reitzner}. In
general dimensions, however, the asymptotic behavior, as $n$ goes to infinity,
is well understood; the general picture that emerges is a dichotomy
between the case where $K$ is a polytope when
\begin{eqnarray}\label{eq:polytope}
  E[f_i(K_n)] &=& c_{d, i, K} \log^{d-1}n + o(\log^{d-1}n),
  \\ \nonumber
  E[V(K)-V(K_n)] &=& c_{d, K} n^{-1} \log^{d-1}n + o(n^{-1}\log^{d-1}n),
\end{eqnarray}
and the case where $K$ is a smooth convex body (i.e. with boundary of differentiability class $C^2$ and with positive Gaussian curvature) when 
\begin{eqnarray}\label{eq:smooth}
  E[f_i(K_n)] &=& c'_{d, i, K} n^{\frac{d-1}{d+1}} + o(n^\frac{d-1}{d+1}),
  \\ \nonumber
  E[V(K)-V(K_n)] &=& c_{d, K} n^{-\frac{2}{d+1}} + o(n^{-\frac{2}{d+1}}).
\end{eqnarray}
(Here $c_{d, i, K}$ and $ c'_{d, i, K}$ are constants depending only
on $i$, $d$ and certain geometric functionals of $K$.)
The results for $E [V(K_n)]$ follow from the corresponding results for $f_0(K_n)$ via Efron's formula~\cite{Efron}.
The literature devoted to such estimates
is consequent and we refer to Reitzner~\cite{Reitzner-survey} for a
recent survey.

\paragraph{Monotonicity of $f$-vectors.}

In spite of numerous papers devoted to the study of random polytopes,
the natural question whether these functionals are monotone remained
open in general.

Concerning the monotonicity of $E [V(K_n)]$ with respect to $n$, the
elementary inequality

\begin{equation}\label{eq:Vmon}
 E[ V (K_{n})] \leq E [V( K_{n+1}) ]
\end{equation}

\noindent
for any fixed $K$ follows immediately from the monotonicity of the
volume.  Yet the monotonicity (at first sight similarly obvious) with
respect to $K$, i.e. the inequality that $K \subset L$ implies

$$ E [V (K_{n})] \leq E[ V( L_{n})] $$

\noindent
surprisingly turned out to be false in general. This problem was posed
by Meckes \cite{MeckesAMS}, see also \cite{Reitzner-survey}, and a
counterexample for $n=d+1$ was recently given by Rademacher
\cite{Rademacher}.

A \lq\lq tantalizing problem\rq\rq\ in stochastic geometry, to quote
Vu~\cite[Section~8]{Vu}, is whether $f_0(K_n)$ is monotone in $n$, 
that is, whether similar to Equation~\eqref{eq:Vmon}:

$$ E [f_0 (K_n)] \leq E [f_0(K_{n+1})] . $$

\noindent
This is not a trivial question as the convex hull of $K_n \cup \{x\}$
has fewer vertices than $K_n$ if $x$ lies outside $K_n$ and sees more
than two of its vertices. Some bounds are known for the expected
number of points of $K_n$ seen by a random point $x$ chosen uniformly
in $K \setminus K_n$~\cite[Corollary~8.3]{Vu} but they do not suffice
to prove that $E[f_0(K_n)]$ is monotone.

It is known that $E[f_0(K_n)]$ is always bounded from below by an
increasing function of $n$, namely $c(d)\log^{d-1}n$ where $c(d)$
depends only on the dimension: this follows, via Efron's
formula~\cite{Efron}, from a similar lower bound on the expected
volume of $K_n$ due to B\'ar\'any and
Larman~\cite[Theorem~2]{Barany-Larman}. While this is encouraging, it
does not exclude the possibility of small oscillations preventing
monotonicity. In fact, B\'ar\'any and
Larman~\cite[Theorem~5]{Barany-Larman} also showed that for any
functions $s$ and $\ell$ such that $\lim_{n \to \infty}s(n)=0$ and
$\lim_{n \to \infty}\ell(n)=\infty$ there exists\footnote{More
  precisely, \emph{most} (in the Baire sense) compact convex sets
  exhibit this behavior.} a compact convex domain $K \subset \R^d$ and
a sequence $(n_i)_{i \in \N}$ such that for all $i\in \N$:

\[ E[f_0(K_{n_{2i}})] > s(n_{2i}) {n_{2i}}^{\frac{d-1}{d+1}} \quad
\hbox{and} 
\quad E[f_0(K_{n_{2i+1}})] <
\ell(n_{2i+1})\log^{d-1}(n_{2i+1}).\]

\noindent
When general convex sets are considered, there may thus be more to
this monotonicity question than meets the eye.

\paragraph{Results.}

This paper present two contributions on the monotonicity of the {\boldmath $f$}-vector of $K_n$.
First, we show that for any \emph{planar} convex body $K$
the expectation of $f_0(K_n)$ is an increasing function of $n$. This
is based on an explicit representation of the expectation $E [f_0(K_n)]$.
\begin{theorem}
\label{theorem:monocity} 
Assume $K$ is a planar convex body. For all integers $n$,  

$$E[f_i(K_{n+1})]>E[f_i(K_n)] ,\ i=0,1 . $$
\end{theorem}

Second, we show that if $K$ is a compact convex set in $\R^d$ with a
$C^2$ boundary then the expectation of $f_{d-1}(K_n)$ is increasing
for $n$ large enough (where \lq\lq large enough\rq\rq depends on $K$);
in particular, for smooth compact convex bodies $K$ in $\R^d$ the
expectation of $f_0(K_n)$ becomes monotone in $n$ for $n$ large
enough.

\begin{theorem}\label{th:smooth}
  Assume $K \subset \R^d$ is a smooth convex body. Then there is an
  integer $n_K $ such that for all $n \geq n_K$

$$E[f_{i}(K_{n+1})]>E[f_{i}(K_n)] ,\  i=d-2, d-1. $$
\end{theorem}

Our result is in fact more general and
applies to convex hulls of points i.i.d. from any \lq\lq sufficiently
generic\rq\rq\  distribution (see Section~\ref{sec:clarkson}) and follows
from a simple and elegant random sampling technique introduced by
Clarkson~\cite{Clarkson} to analyze \emph{non-random} geometric
structures in discrete and computational geometry.

\section{Monotonicity for convex domains in the plane}
Assume $K$ is a planar convex body of volume one. 
Given a unit vector $u \in \mathbb{S}^1$, each halfspace 
$\{ x \in \R^2:\ x \cdot u \leq p \}$ cuts off from $K$ a set of volume $s\in [0,1]$. On the contrary, given $u \in \mathbb{S}^1, s \in (0,1)$ there is a unique line $\{ x \in \R^2:\ x \cdot u = p \}$ and a corresponding halfspace $\{ x \in \R^2:\ x \cdot u \leq p \}$ which cuts of from $K$ a set of volume precisely $s$. Denote by $L(s,u)$ the square of the length of this unique chord.

Using this notation, Buchta and Reitzner \cite{Buchta-Reitzner} showed that 
the expectation of the number of points on the convex hull of
$n$ points chosen randomly uniformly in $K$ can be computed with the following formula.

\begin{lemma}[Buchta,Reitzner]
$$ E[f_0(K_n)]  =
\frac16  n(n-1)
         \int_{\mathbb{S}^1} \int_{0}^1  s^{n-2}
                                         L(s,u) \, ds du  $$
\end{lemma}

Here $du$ denotes integration with respect to Lebesgue measure on
$\mathbb{S}^1$.  By a change of variables we obtain

$$ E[f_0(K_n)]  =
\frac16  (n-1)
\int_{\mathbb{S}^1} \int_{0}^1  t^{-\frac{1}n } L(t^{\frac 1n},u) \, dt du  
= 
\frac16 \int_{\mathbb{S}^1} \cI_n(u) du  
$$

\noindent
with $ \cI_n(u) = (n-1) \int_{0}^1 t^{-\frac{1}n } L(t^{\frac 1n},u)
\, dt .$ Observe that $L(1,u)=0$ for almost all $u \in \mathbb{S}^1$.
Also observe that the partial derivative $\frac{\partial }{\partial s}
L(s,u)$ exists for almost all $(s,u)$. This is a consequence of the
a.e. differentiability of a convex function. In the following we write
$L(\cdot)=L(\cdot, u)$, ${\textstyle \frac{\partial}{\partial s}}
L(\cdot,u) = L'(\cdot)$.  Integration by parts in $t$ gives

\begin{equation}\label{eq:cI} \cI_n (u) =
- \int_{0}^1  L'(t^{\frac 1n}) \, dt .
\end{equation}

\noindent
Finally, the convexity of $K$ induces the following lemma.
\begin{lemma}
\label{lemma:convex} Given a value $u$, the derivative 
$s\mapsto L'(s)$ is a decreasing function.
\end{lemma}

\begin{proof}
Fix $u \in \mathbb{S}^1$. We denote $l(p)=l(p,u)$ the length of the chord $K \cap \{ x \in \R^2:\ x \cdot u = p \}$.
Moreover,
$L(s(p))=l(p)^2$ where $s(p)$ is the volume of the part of $K$ on the
left of the chord of length $l(p)$. 
This volume $s(p)$ is a monotone and hence injective function of $p$ with inverse $p(s)$, and we have 
$\frac{d}{d p} s(p) = l(p)$.

First observe that because $K$ is a convex body, the chord length
$l(p)$ is concave as a function of $p$. 
Thus its derivative $\frac{d}{d p} l(p)$ is decreasing.

Hence
$$L'(s) = 2 l(p(s)) \frac{d}{ds} l(p(s))
= 2  \frac{d}{ds} l(p(s)) \  (\frac{d}{dp} s(p))\vert_{p=p(s)}
= 2 \frac{d}{dp} l(p)\vert_{p=p(s)}.$$
Since the derivative $\frac{d}{dp} l(p)$ is decreasing and $p(s)$ is monotone, we see that ${\textstyle \frac{d}{ds}} L(s,u)$ is decreasing in $s$.
\end{proof}

This allows us to conclude that the expectancy of the number of points
on the convex hull is increasing.

\begin{proof}[Proof of Theorem~\ref{theorem:monocity}]
  According to Lemma \ref{lemma:convex}, $L'(s)$ is decreasing. Since
  for all $t\in[0,1]$, $t^\frac1n < t^{\frac1{n+1}}$, this implies
  that $L'(t^\frac1{n}) \geq L'(t^\frac1{n+1})$. Combined with
  equation~(\ref{eq:cI}) we have

$$ \cI_n (u) \leq \cI_{n+1} (u) $$

\noindent
which proves $E[f_0(K_n)] \leq E[f_0(K_{n+1})]$. In the planar case,
the number of edges is also an increasing function because
$f_0(K_n)=f_1(K_n)$.
\end{proof}

\section{Random sampling}\label{sec:clarkson}

We denote by $\sharp S$ the cardinality of a finite set $S$ and we let
${\mathbb 1}_X$ denote the characteristic function of event $X$:
${\mathbb 1}_{p\in F}$ is 1 if $p\in F$ and $0$ otherwise. Let $S$ be
a finite set of points in $\R^d$ and let $k\ge 0$ be an integer. A
\emph{$k$-set of $S$} is a subset $\{p_1, p_2, \ldots p_{d}\}
\subseteq S$ that spans a hyperplane bounding an open halfspace that
contains exactly $k$ points from $S$; we say that the $k$-set
\emph{cuts off} these $k$ points. In particular, $0$-sets are facets
of the convex hull of $S$. For any finite subset $S$ of $D$ we let
$s_k(S)$ denote the number\footnote{If the hyperplane separates $S$ in
  two subsets of $k$ elements ($2k+d=n$) the $k$-set is counted
  twice.} of $k$-sets of $S$.

\paragraph{Generic distribution assumption.}
  Let $\pr$ denote a probability distribution on $\R^d$; we assume
  throughout this section that $\pr$ is such that $d$ points chosen
  independently from $\pr$ are generically affinely independent.

\bigskip

Determining the order of magnitude of the maximum number of $k$-sets
determined by a set of $n$ points in $\R^d$ has been an important open
problem in discrete and computational geometry over the last decades.
In the case d=2, Clarkson~\cite{Clarkson} gave an elegant proof that
$s_{\le k}(S)=O(nk)$ for any set $S$ of $n$ points in the plane,
where:

\[ s_{\le k}(S)=s_0(S)+s_1(S)+ \ldots + s_k(S).\]

\noindent
(See also Clarkson and Shor~\cite{Clarkson-Shor} and
Chazelle~\cite[Appendix A.2]{Chazelle}.) Although the statement holds
for any fixed point set, and not only in expectation, Clarkson's
argument is probabilistic and will be our main ingredient. It goes as
follows. Let $R$ be a subset of $S$ of size $r$, chosen randomly and
uniformly among all such subsets.  A $i$-set of $S$ becomes a $0$-set
in $R$ if $R$ contains the two points defining the $i$-set but none of
the $i$ points it cuts off. This happens approximately with
probability $p^2(1-p)^i$, where $p=\frac{r}n$ (see
\cite{Clarkson-Shor} for exact computations). It follows that:

\[ E[s_0(R)] \ge \sum_{0 \le i \le k} p^2(1-p)^is_i(S) \ge p^2(1-p)^k
s_{\le k}(S).\]

\noindent
Since $E[s_0(R)]$ cannot exceed $\sharp R = r$, we have $s_{\le k}(S)
\le \frac{n}{p(1-p)^k}$ which, for $p=1/k$, yields the announced bound
$s_{\le k}(S)=O(nk)$.  A similar random sampling argument yields the
following inequalities.

\begin{lemma}\label{lem:clarkson}
  Let $s_k(n)$ denote the expected number of $k$-sets in a set of $n$
  points chosen independently from $\pr$. We have

  \begin{equation}
    s_0(n) \ge s_0(n-1)+\frac{d\ s_0(n)-s_1(n)}n\label{eq:C1}
  \end{equation}

\noindent
  and for any integer $1 \le r \le n$ 

  \begin{equation}
    s_0(r) \ge \frac{\binom{n-d}{r-d}}{\binom{n}{r}} s_0(n)
  + \frac{\binom{n-d-1}{r-d}}{\binom{n}{r}} s_1(n).\label{eq:C2}
  \end{equation}
\end{lemma}
 
\begin{proof}
  Let $S$ be a set of $n$ points chosen, independently, from $\pr$ and
  let $q \in S$. The $0$-sets of $S$ that are not $0$-sets of $\Sp$
  are precisely those defined by $q$. Conversely, the $0$-sets of
  $\Sp$ that are not $0$-sets of $S$ are precisely those $1$-sets of
  $S$ that cut off $q$. We can thus write:

  \[ s_0(S) = s_0(\Sp) + \sum_{F \hbox{ facet of } CH(S)} \
  {\mathbb{1}}_{q \in F} - \sharp
  1\mbox{-sets cutting off } q.\]

\noindent
  Note that the equality remains true in the degenerate cases where several
  points of $S$ are identical or in a non generic position if we
  count the facets of the convex hull of $S$ with multiplicities in the
  sum and in $s_0(n)$.
%
  Summing the previous identity over all points $q$ of $S$, we obtain

  \[ns_0(S) = \pth{\sum_{q \in S}s_0(\Sp)} 
    + \pth{\sum_{F \hbox{ facet of } CH(S)} \sum_{q \in S} {\mathbb{1}}_{q \in F}}
     - s_1(S),\]

\noindent
  and since a face of $CH(S)$ has at least $d$ vertices,

  \[ns_0(S) \geq \pth{\sum_{q \in S}s_0(\Sp)} + d s_0(S) - s_1(S).\]

\noindent
  This inequality is actually an equality if $d=2$ but not if $d \ge
  3$.
  Taking~$n$ points chosen randomly and independently from $\pr$, then
  deleting one of these points, chosen randomly with equiprobability,
  is the same as taking $n-1$ chosen randomly and independently from
  $\pr$. We can thus average over all choices of $S$ and obtain

  \[ns_0(n) \geq ns_0(n-1)+d\ s_0(n)-s_1(n),\]

\noindent
  which implies Inequality~\eqref{eq:C1}.
  

  Now, let $r \le n$ and let $R$ be a random subset of $S$ of size
  $r$, chosen uniformly among all such subsets. For $k \le
  \frac{r}2 $, a $k$-set of $S$ is a $0$-set of $R$ if $R$
  contains that $k$-set and none of the $k$ points it cuts off. For
  each fixed $k$-set $A$ of $S$, there are therefore
  $\binom{n-d-k}{r-d}$ distinct choices of $R$ in which $A$ is a
  $0$-set. Counting the expected number of $0$-sets and $1$-sets of
  $S$ that remain/become a $0$-set in $R$, and ignoring those $0$-sets
  of $R$ that were $k$-sets of $S$ for $k \ge 2$, we obtain:

  \[ E[s_0(R)] \ge \frac{\binom{n-d}{r-d}}{\binom{n}{r}} s_0(S)
  + \frac{\binom{n-d-1}{r-d}}{\binom{n}{r}} s_1(S).\]

\noindent
  Recall that the expectation is taken over all choices of a subset
  $R$ of $r$ elements of the fixed point set $S$. We can now average
  over all choices of a set $S$ of $n$ points taken, independently,
  from $\pr$, and obtain Inequality~\eqref{eq:C2}.
\end{proof}


Letting $p=\frac{r-d}{n-d}$, Inequality~\eqref{eq:C2} yields the,
perhaps more intuitive, inequality:

\[ s_0\pth{r} \ge p^d s_0(n) + p^d(1-p)s_1(n).\]

\noindent
 We can now prove our main result.

\begin{theorem}\label{thm:clarkson}
  Let $Z_n$ denote the convex hull of $n$ points chosen independently
  from $\pr$. If $E[f_{d-1}(Z_n)] \approx A n^{c}$ for some $A$, $c>0$
  then there exists an integer $n_0$ such that for any $n \ge n_0$ we
  have $E[f_{d-1}(Z_{n+1})] > E[f_{d-1}(Z_n)]$.
\end{theorem}
\begin{proof}[Proof of Theorem~\ref{thm:clarkson}]
  Let $\pr$ be a probability distribution on $\R^d$ and let $s_k(n)$
  denote the expected number of $k$-sets in a set of $n$ points chosen
  independently from $\pr$. Recall that  $s_0$ counts the expected
  number of facets in the convex hull of $n$ points chosen
  independently and uniformly from $\pr$. By Inequality~\eqref{eq:C1},
  for $s_0$ to be increasing it suffices that $s_1(n)$ be bounded from
  above by $d s_0(n)$.

  Let $\psi(r,n)=\frac{s_0(r)}{s_0(n)}$. Substituting into
  Inequality~\eqref{eq:C2},

  \[ \psi(r,n) s_0(n) \ge \frac{\binom{n-d}{r-d}}{\binom{n}{r}} s_0(n)
  + \frac{\binom{n-d-1}{r-d}}{\binom{n}{r}} s_1(n),\]

\noindent
  which rewrites as

  \begin{equation}
  \label{eq:s1s0}
    s_1(n) \le \frac{\binom{n-d}{r-d}}{\binom{n-d-1}{r-d}}
    \pth{\psi(r,n)
      \frac{\binom{n}{r}}{\binom{n-d}{r-d}} - 1}s_0(n)
  \end{equation}

  We let $q=\frac{n-d}{r-d}$. Developing the
  binomial expressions:

  \begin{align*}
    \frac{\binom{n-d}{r-d}}{\binom{n-d-1}{r-d}} &= \frac{n-d}{n-r}
                                                 = \frac{q}{q-1} &
    \text{and}&&
    \frac{\binom{n}{r}}{\binom{n-d}{r-d}} &=
          \frac{n}{r}\cdot\frac{n-1}{r-1}\ldots\frac{n-d+1}{r-d+1}\\
  \end{align*}
  
  And for $0 \leq k < d$ we have $\frac{n-k}{r-k}<\frac{n-d}{r-d}=q$.
  Thus:

  $$s_1(n) \leq q \frac{\psi(r,n)\frac{n}{r}q^{d-1}-1}{q-1} s_0(n)$$

  Assume now that we know a function $g$ such that $s_0(n)\approx
  g(n)$.  Then for any $\frac{1}{2}>\epsilon>0$, there is
  $N_\epsilon\in\mathbb{N}$ such that for all $n>r>N_\epsilon$ we
  have:

  $$\psi(r,n)=\frac{s_0(r)}{s_0(n)}
  <\left(\frac{1+\epsilon}{1-\epsilon}\right)\frac{g(r)}{g(n)}
  <(1+4\epsilon)\frac{g(r)}{g(n)}$$
  
  \noindent
  which gives:

  \begin{equation}
    s_1(n) \leq q \frac{\frac{g(r)/r}{g(n)/n}q^{d-1}-1}{q-1} s_0(n)
    + 4\epsilon\left(\frac{g(r)/r}{g(n)/n}\right)
    \frac{q^{d}}{q-1} s_0(n).
    \label{eq:general}
  \end{equation}

  \noindent
  In the case where $s_0(n)\approx A n^c$, we have:

  $$ \frac{g(r)/r}{g(n)/n} = \left(\frac{n}{r}\right)^{1-c} < q^{1-c} $$

  \noindent
  And plugging back in Equation \eqref{eq:general}, we get:
                        
  \begin{align}
       s_1(n) &\leq q \frac{q^{d-c}-1}{q-1} s_0(n)
                + 4\epsilon \frac{q^{d+1}}{q-1} s_0(n)
    \label{eq:twoterms}
  \end{align}

  The expression $q \frac{q^{d-c}-1}{q-1}$ converges toward $d-c$ when
  $q$ approaches $1$. Thus, there exists $\epsilon_d$ such that for
  all $1<q<1+\epsilon_d$:

  $$q \frac{q^{d-c}-1}{q-1}<d-\frac{c}{2}$$

  \noindent
  The second term of Equation \eqref{eq:twoterms} is bounded for all
  $1+\frac{\epsilon_d}{2}<q< 1+\epsilon_d$ by:

  $$4\epsilon\frac{q^{d+1}}{q-1}s_0(n)
  <8\epsilon\frac{(1+\epsilon_d)^{d+1}}{\epsilon_d}s_0(n)$$

  Finally, let $\epsilon= \frac{c\epsilon_d}{32(1+\epsilon_d)^{d+1}}$.
  For all $r$ such that:

  \begin{align}
    N_\epsilon&<r<n &\text{and}&&
    1+\frac{\epsilon_d}{2}&<\frac{n-d}{r-d}<\epsilon_d
    \label{eq:Cr}
  \end{align}

  \noindent
  we have:

  $$s_1(n)<(d-\frac{c}{4})s_0(n)$$

  \noindent
  With Inequality~\eqref{eq:C1} this implies that $s_0(n)>s_0(n-1)$.

  It remains to check that for $n$ large enough, there always exists
  $r$ satisfying Condition \eqref{eq:Cr}. We can rewrite condition
  \eqref{eq:Cr} as:

  $$d+\frac{n-d}{1+\epsilon_d}<r<d+\frac{n-d}{1+\epsilon_d/2}$$

  \noindent
  In particular, there exists an integer $r$ satisfying this condition
  as soon as $\frac{n-d}{1+\epsilon_d/2}-\frac{n-d}{1+\epsilon_d}>1$.
  Thus, as soon as $n>max(N_\epsilon,
  d+\frac{1}{\frac{1}{1+\epsilon_d/2}-\frac{1}{1+\epsilon_d}})$,
  Condition \eqref{eq:Cr} is satisfied, which concludes the proof.
\end{proof}

Now, from equations~(\ref{eq:polytope}) and~(\ref{eq:smooth}) we can see
that Theorem~\ref{thm:clarkson} holds for random polytopes $K_n$ when
$K$ is smooth, but not when $K$ is a polytope. This proves that for smooth $K$ the expectation $E[f_{d-1}(K_n)]$ is asymptotically
 increasing, i.e. the first part of Theorem~\ref{th:smooth}.
  
The genericity assumption on $\pr$ implies that the convex hull $Z_n$
of $n$ points chosen independently from $\pr$ is almost surely
simplicial. Thus, in $Z_n$ any $(d-1)$-face is almost surely incident
to exactly $d$ faces of dimension $d-2$ and
$f_{d-2}(Z_n)=\frac{d}{2}f_{d-1}(Z_n)$. Theorem~\ref{thm:clarkson} 
therefore implies that $f_{d-2}(Z_n)$ is asymptotically increasing, i.e. the second part of Theorem~\ref{th:smooth}..
For $d=3$, with Euler's relation this further implies that $f_0(Z_n)$
is asymptotically increasing.

\section*{Acknowledgment.} 

This work was initiated during the $10^{th}$ McGill - INRIA Workshop
on Computational Geometry at the Bellairs institute. The authors wish
to thank all the participants for creating a pleasant and stimulating
atmosphere, in particular Sariel Har-Peled and Raimund Seidel for
useful discussions at the early stage of this work.  The authors also
thank Imre B\'ar\'any and Pierre Calka for helpful discussions.

\end{document}